\pdfoutput=1
\RequirePackage{ifpdf}
\ifpdf % We are running pdfTeX in pdf mode
\documentclass[pdftex]{sigma}
\else
\documentclass{sigma}
\fi

\def\bs{\boldsymbol}

\numberwithin{equation}{section}

\newtheorem{Theorem}{Theorem}[section]
 { \theoremstyle{definition}
\newtheorem{Notation}[Theorem]{Notation}
\newtheorem{Example}[Theorem]{Example}
\newtheorem{Remark}[Theorem]{Remark}
\newtheorem{Remarks}[Theorem]{Remarks}
 }

\begin{document}

%\allowdisplaybreaks

\renewcommand{\thefootnote}{$\star$}

\newcommand{\arXivNumber}{1512.07833}

\renewcommand{\PaperNumber}{013}

\FirstPageHeading

\ShortArticleName{Separability in Riemannian Manifolds}

\ArticleName{Separability in Riemannian Manifolds\footnote{This paper is a~contribution to the Special Issue
on Analytical Mechanics and Dif\/ferential Geometry in honour of Sergio Benenti.
The full collection is available at \href{http://www.emis.de/journals/SIGMA/Benenti.html}{http://www.emis.de/journals/SIGMA/Benenti.html}}}

\Author{Sergio BENENTI}

\AuthorNameForHeading{S.~Benenti}

\Address{Dipartimento di Matematica, Universit\`a di Torino,  via Carlo Alberto 10, 10123 Torino, Italy}
\Email{\href{mailto:sergio.benenti@unito.it}{sergio.benenti@unito.it}}
\URLaddress{\url{http://www.sergiobenenti.it/}}

\ArticleDates{Received December 27, 2015, in f\/inal form January 19, 2016; Published online February 01, 2016}

\Abstract{An outline of the basic Riemannian structures underlying the separation of variables in the Hamilton--Jacobi equation of natural Hamiltonian systems.}

\Classification{70H20; 35F21; 58B21}

\begin{flushright}
\begin{minipage}{150mm}
{\it This paper was submitted to the Royal Society and accepted for publication in a~special volume dedicated to the State of the Art of the Separation of Variables $($2004$)$. However, this volume was never published due to the sudden death of the Editor, Vadim Kutznetsov, to whom I dedicate this work.}
\end{minipage}
\end{flushright}

\renewcommand{\thefootnote}{\arabic{footnote}}
\setcounter{footnote}{0}

\section{Introduction}\label{section1}

In the past years the study of {\it separable systems} (whose Hamilton--Jacobi equations can be integrated by separation of variables) has shown a remarkable development, also in relation with other kinds of integrable systems (bi-Hamiltonian systems, Lax systems). I~think that an outline of the Riemannian background of this theory may be useful for specialists as well as beginners. With the exception of the last section, we will conf\/ine our discussion to the {\it orthogonal separable systems} (called {\it St\"ackel systems}) and to a special class of St\"ackel systems, referred to as {\it L-systems}. Some of the theorems presented here are new. For some of the recalled theorems a new shorter proof is provided.
Let $Q$ be an $n$-dimensional Riemannian manifold with generic local coordinates $\underline q=(q^i)$ and (contravariant) metric tensor $\bs G=(g^{ij})$, which we assume to be positive-def\/inite, and let $T^*Q$ be the cotangent bundle of $Q$, with canonical coordinates $(\underline q,\underline p)=(q^i,p_i)$. We will deal with the additive separation of the Hamilton--Jacobi equations
\begin{gather*}
G(\underline q,\underline p)=E, \qquad
H(\underline q,\underline p)=E, \qquad  p_i=\partial_iW,
\end{gather*}
where $G=\tfrac 12 g^{ij} p_ip_j$ is the {\it geodesic Hamiltonian} on $T^*Q$, and   $H=G+V=\tfrac 12 g^{ij} p_ip_j+V$ is a~{\it natural Hamiltonian}, $V(\underline q)$  being the {\it potential energy}, a smooth function on $Q$ canonically lifted to a function on~$T^*Q$. A coordinate system $\underline q$ is called {\it separable} if the geodesic Hamilton--Jacobi equation $G=E$ admits a {\it complete solution} of the form
\begin{subequations}\label{(1.1)}
\begin{gather}
W(\underline q,\underline c)=\sum _{i=1}^n W_i\big(q^i,\underline c\big),\qquad \underline c=(c^i), \label{(1.1)1}\\
\det\left[\dfrac{\partial^2W}{\partial q^i\partial c^j}\right]\neq 0.\label{(1.1)2}
\end{gather}
\end{subequations}

Such a solution is called {\it separated solution}. Also note, that in these def\/initions the presence of a set of $n$ constants $\underline c$ satisfying the {\it completeness condition} \eqref{(1.1)2} is fundamental. Note that here we consider only {\it natural canonical coordinates}, where $\underline q$ are coordinates on the conf\/iguration manifold. Levi-Civita~\cite{LC-1904} proved that {\it a Hamilton--Jacobi equation $H(\underline q,\underline p)=E$ admits a~separated solution \eqref{(1.1)} if and only if the differential equations\footnote{Notation: $\partial_i=\partial/\partial q^i$, $\partial^i=\partial/\partial p_i$.}
\begin{gather*}
L_{ij}(H)\doteq \partial _i H\partial_jH \partial ^i\partial^jH+
\partial^iH \partial ^j \partial _i\partial_jH-
\partial_iH \partial ^jH \partial ^i\partial _jH-
\partial^iH \partial _jH \partial _i\partial ^jH=0
\end{gather*}
are identically satisfied}. These are known as the {\it separability conditions} or {\it separability equations of Levi-Civita}. They provide not only a~simple method for testing whether a~coordinate system is separable or not, but also the basis for the geometrical (i.e., intrinsic) characterisation of the separation. A f\/irst (and well known) example of application is the following: the Levi-Civita equations for a natural Hamiltonian, $L_{ij}(G+V)=0$, are polynomial equations of fourth degree in the momenta $\underline p$, which must be identically satisf\/ied for all admissible values of these variables. It is easy to note that the fourth-degree homogeneous part of these equations is~$L_{ij}(G)=0$. This means that: (i)~the separation of the geodesic equation is a~necessary condition for the separation of equation $G+V=E$; (ii)~the study of the geodesic separation plays a prominent role, (iii) the above-given def\/inition of {\it separable coordinates} makes sense.

A special but fundamental case in this theory is the {\it orthogonal separation}, where the coordinates are assumed to be orthogonal, $g^{ij}=0$ for $i\neq j$. In this case, examined f\/irstly by St\"ackel~\cite{Stackel-1893,Stackel-1897}, later on by Levi-Civita~\cite{LC-1904}, Eisenhart~\cite{Eisenhart-1934,Eisenhart-1949}, and more recently by many authors, the Levi-Civita equations $L_{ij}(G)=0$ are equivalent to equations
\begin{gather}\label{(1.2)}
S_{ij}\big(g^{kk}\big)=0, \qquad i\neq j,
\end{gather}
where $S_{ij}(\cdot)$ denote the {\it St\"ackel operators} associated with an orthogonal metric $(g^{ii})$.
For any smooth function $V$ on $Q$, it is def\/ined by
\begin{gather*}
S_{ij}(V)\doteq \partial_i\partial_j V-\partial_i\ln g^{jj}\partial_j V-\partial_j\ln g^{ii}\partial_i V,  \qquad i\neq j.
\end{gather*}
The Levi-Civita equations $L_{ij}(G+V)=0$ are equivalent to
\begin{gather}
S_{ij}\big(g^{kk}\big)=0, \qquad S_{ij}(V)=0.
\label{(1.3)}
\end{gather}

\section{Killing tensors}\label{section2}

As shown by Eisenhart \cite{Eisenhart-1934,Eisenhart-1949} (for the orthogonal case) and by Kalnins and Miller \cite{KM-1980,KM-1981}, the geodesic separation is related to the existence of Killing vectors and Killing tensors of order two. In this section we recall the basic properties of these objects.
The contravariant symmetric tensors $\bs K=(K^{i\ldots j})$ on $Q$ are in one-to-one correspondence with homogenous polynomials on~$T^*Q$,
\begin{gather*}
{\bs K}=\big(K^{i\ldots j}\big)\ \longleftrightarrow \  P_{{\bs K}}=P(\bs K)=K^{i\ldots j} p_i\cdots p_j.
\end{gather*}
For a tensor of order zero, i.e., a function $f$ on $Q$, we def\/ine $P_f\doteq f$, where $f$ is canonically lifted to $T^*Q$ (by constant values on the f\/ibers). The space of these polynomial functions is closed with respect to  the canonical Poisson bracket
\begin{gather*}
\{A,B\}\doteq \partial^iA \partial_iB-\partial^iB \partial_iA.
\end{gather*}
Hence, on the space of the symmetric contravariant tensors we def\/ine a Lie-algebra structu\-re~$[\cdot,\cdot]$ by setting
\begin{gather*}
P([{\bs K}_1,{\bs K}_2])=\{P({\bs K}_1),P({\bs K}_2)\},
\end{gather*}
and the {\it symmetric product} $\odot$ by setting
\begin{gather*}
P({\bs K}_1\odot{\bs K}_2)=P({\bs K}_1)\cdot P({\bs K}_2).
\end{gather*}
Note that all the above-given def\/initions do not depend on a metric tensor.
If a metric tensor~$\bs G$ is present, then we say that $\bs K$ is a {\it Killing tensor} (KT) if $P(\bs K)$ is in involution with $P(\bs G)=2G$,
\begin{gather}
\{P(\bs K),P(\bs G)\}=0\ \Longleftrightarrow\  [\bs K,\bs G]=0.
\label{(2.1)}
\end{gather}
This means that $P(\bs K)$ is a f\/irst integral of the geodesic f\/low.
In the special case of a function~$f$, this def\/inition is equivalent to $\nabla f=0$ (by $\nabla f$ we denote the {\it gradient} of a function~$f$). A~vector f\/ield~$\bs X$ is a {\it Killing vector}, $[\bs X,\bs G]=0$, if and only if its f\/low preserves the metric.

Let us consider the case of a symmetric 2-tensor $\bs K$. Since a metric tensor is present, the boldface object~$\bs K$ can be represented in components as a tensor of type $(2,0)$, $(1,1)$ and $(0,2)$, respectively, $\bs K= (K^{ij})=(K^i_j)=(K_{ij})$.

As a symmetric tensor of type $(1,1)$, $\bs K$ def\/ines an endomorphism on the space $\cal X(Q)$ of the (smooth) vector f\/ields on~$Q$ and an endomorphism on the space $\Phi^1(Q)$ of the (smooth) 1-forms on $Q$. We will denote by
$\bs K\bs X$ the vector f\/ield image of $\bs X\in\cal X(Q)$ by $\bs K$, and by $\bs K\phi$ the 1-form
image of $\phi\in\Phi^1(Q)$ by $\bs K$. This means that $\bs K\bs X=K^i_jX^j \partial_i$, $\bs K\phi= K^i_j\phi_i\,dq^j$.
Note that the metric tensor $\bs G$ coincide with the identity operator $\bs I$,
whose $(1,1)$ components are given by the Kronecker symbol $\delta^i_j$.
Then a 2-tensor $\bs K$ gives rise to eigenvalues, eigenvectors or eigenforms, according to equations
$\bs K\bs X=\rho\bs X$, $\bs K\phi=\rho \phi$. We recall that, in a positive-def\/inite metric, (i)~all symmetric tensors have real eigenvalues; (ii)~the {\it algebraic multiplicity} of an eigenvalue $\rho$ (i.e., its order as a root of the characteristic equation $\det(\bs K-\rho\bs G)=0$) is equal to its {\it geometrical multiplicity}
(i.e., the dimension of the space of the corresponding eigenvectors, or eigenforms); the eigenspaces corresponding to distinct eigenvalues are orthogonal. We will denote by $\bs K_1\bs K_2$ the product of the two endomorphisms $\bs K_1$ and $\bs K_2$; in components $(\bs K_1\bs K_2)^{ij}=K_1^{ih}K_{2h}^{\;j}$. The {\it algebraic commutator} of the two tensors will be denoted by
\begin{gather*}
[ \! [\bs K_1,\bs K_2] \! ]\doteq \bs K_1\bs K_2-\bs K_2\bs K_1.
\end{gather*}
If a symmetric 2-tensor $\bs K$ can be diagonalised in orthogonal coordinates, $K^{ij}=0$ for $i\neq j$, then $K^{ii}=\rho^i g^{ii}$, where $(\rho^i)$ are the eigenvalues of $\bs K$. By writing the {\it Killing equation}~\eqref{(2.1)} in these coordinates, we see that {\it $\bs K$ is a KT if and only if
equations
\begin{gather}
\partial_i\rho^j=\big(\rho^i-\rho^j\big) \partial_i \ln g^{jj}.
\label{(2.2)}
\end{gather}
are satisf\/ied by the eigenvalues}. These equations have been called  {\it Killing--Eisenhart equations} in \cite{BCR-2002a} since they have been extensively used by Eisenhart~\cite{Eisenhart-1949}. However, they appear in an earlier paper by Levi-Civita \cite[p.~285]{LC-1896}.

\section {Killing--St\"ackel spaces}\label{section3}

Equations \eqref{(2.2)} can be interpreted as a linear system of~$n$ f\/irst-order partial dif\/ferential equations in normal form, in the $n$ unknown functions $\rho^i(\underline q)$. It is a remarkable fact that the integrability conditions assume of the form
\begin{gather*}
\big(\rho^i-\rho^j\big) S_{ij}\big(g^{kk}\big)=0.
\end{gather*}
Then their link with the orthogonal separation is at once clear. A second, and even more remarkable property, is that the unknown functions~$\rho^i$ appear in the integrability conditions through their dif\/ferences $\rho^i-\rho^j$. This means that
{\it if the system~\eqref{(2.2)} admits a solution such that $\rho^i\neq\rho^j$, then it is completely integrable}. Note that for a linear system the converse is always true. Going back to the orthogonal separability conditions~\eqref{(1.2)} we can immediately conclude that: (I)~{\it a system of orthogonal coordinates is separable if and only if there exists a KT which is diagonalised in these coordinates and which has pointwise simple eigenvalues}. Furthermore, since the system~\eqref{(2.2)} is linear, if it is completely integrable then it admits a~$n$-dimensional space of solutions (and its converse). As a~consequence: (II)~{\it a~Killing tensor~$\bs K$ which has simple eigenvalues and is diagonalised in orthogonal coordinates generates a~$n$-dimensional spa\-ce~$\cal K$ of Killing tensors which are all diagonalised in the same coordinates}. Such a space will be called {\it Killing--St\"ackel space} (KS-space). In the space of direct sums of Killing tensors, ${\bs K}=c \oplus \bs K_1\oplus \bs K_2\oplus \cdots \oplus \bs K_n\oplus\cdots$ endowed with the Lie bracket $[\cdot,\cdot]$ def\/ined above, a~KS-space, which is made of elements $0\oplus 0\oplus \bs K_2\oplus 0\oplus\cdots$, is an involutive subalgebra. For this reason it has also been called {\it Killing--St\"ackel algebra} in~\cite{BCR-2002a,BCR-2002b}.

Three remarks are in order: (i)~{\it the metric tensor belongs to any KS-space}~-- indeed, $\rho^i=1$ is a trivial solution of the system \eqref{(2.2)}; (ii)~{\it if two KS-spaces have an element with simple eigenvalues in common,
then they coincide}; (iii)~{\it all elements of a KS-space are in involution}
(if equations~\eqref{(2.2)} are satisf\/ied for two tensors $\bs K_1$ and $\bs K_2$, then $\{P(\bs K_1),P(\bs K_2)\}=0$).

All the above properties have a local character and are related to a coordinate system. We remark, however, that they are more precisely related to an equivalence class of orthogonal systems, being equivalent to two systems of coordinates $\underline q$ and $\underline q'$ simply related by a {\it separated transformation} or a {\it rescaling}: $q^i=q^i(q^{i'})$.

We look for a coordinate-free description of all this matter. To this end we recall some basic concepts.

A {\it frame} on a dif\/ferentiable manifold $Q$ (not necessarily Riemannian) is a set of vector f\/ields~$(\bs X_i)$ which form a basis of the tangent space $T_qQ$ at each point $q$ of their domain of def\/inition. In general, frames exist only locally. Global frames are def\/ined if and only if the manifold is parallelisable, i.e., when $TQ\simeq Q\times \mathbb R^n$. Two frames $(\bs X_i)$ and $(\bs X'_i)$ are said to be {\it equivalent} if there are nowhere vanishing functions $f_i$ such that $\bs X_i=f_i\bs X'_i$. A frame is called {\it holonomic} or {\it integrable} if it is equivalent to a {\it natural frame} $(\partial_i)$ associated with coordinates~$(q^i)$.
A basic property is  (cf.\ \cite{Schouten-1954} and \cite{BCR-2002a})

\begin{Theorem}\label{t:3.1}
The three following conditions are equivalent:
$(i)$~the frame $(\bs X_i)$ is holonomic, $(ii)$~for each pair of indices $(i,j)$ the distribution spanned by the vectors $\bs X_i$ and $\bs X_j$ is completely integrable, $(iii)$~for each index $i$ the distribution spanned by the $n-1$ vectors $\bs X_j$ for $j\neq i$ is completely integrable.
\end{Theorem}

On a Riemannian manifold a vector f\/ield $\bs X$ is called {\it normal} if it is {\it orthogonally integrable} or {\it surface forming}, i.e., if it is orthogonal to a family of hypersurfaces. In a positive-def\/inite metric a symmetric tensor $\bs K$ with simple eigenvalues and normal eigenvectors gives rise to and equivalence class of holonomic orthogonal frames hence, to an equivalence class of orthogonal coordinates. Then we get the following simple intrinsic characterisation of the orthogonal geodesic separation \cite[Corollary~4, Section~3]{KM-1980} (see also~\cite{Benenti-1993}):

\begin{Theorem} \label{t:3.2}
The geodesic Hamilton--Jacobi equation is separable in orthogonal coordinates if and only if there exists a Killing $2$-tensor with simple eigenvalues and normal eigenvectors.
\end{Theorem}

A KT having these properties will be called a {\it characteristic Killing tensor} (ChKT). The eigenvectors generates a family of~$n$ orthogonal foliations  of manifolds of codimension~1, which we call {\it St\"ackel web}. Any coordinate system~$(q^i)$ such that the web is locally described by equations $q^i={\rm const}$ (this is equivalent to say that~$dq^i$ are eigenforms of the ChKT),
is separable.

In accordance with the remarks above, a ChKT generates a KS-space. A suitable coordinate-independent def\/inition of this concept is the following: a {\it St\"ackel space} on a Riemannian mani\-fold~$Q_n$ is a $n$-dimensional linear space ${\cal K}_n$ of Killing 2-tensors whose elements (a) {\it commute as linear operators}, $[ \! [\bs K_1,\bs K_2] \! ]=0$, and (b) {\it are in involution}, $[\bs K_1,\bs K_2]=0$. Indeed, in the algebraic realm it can be proved that in such a space there exists an element with pointwise distinct eigenvalues (in the neighborhood of any given point of the domain of def\/inition of ${\cal K}_n$); as a consequence, the commutation relation (a), applied to such a tensor $\bs K_1$, shows that all elements have common eigenvectors. Furthermore, from (b) it follows that

\begin{Theorem}\label{t:3.3}
If $n$ independent KT's in involution have the same eigenvectors, then these eigenvectors are normal.
\end{Theorem}

This remarkable property was f\/irstly discovered by Kalnins and Miller \cite{KM-1980}. However, it is also a remarkable fact that in this last theorem the assumption that the independent tensors are KTs is redundant. In fact, it can be proved that

\begin{Theorem}\label{t:3.4}
An orthogonal frame made of common eigenvectors of $n$ independent symmetric $2$-tensors in involution is holonomic $($the eigenvectors are normal$)$.
\end{Theorem}

For a detailed discussion and proof see \cite{BCR-2002a}. As a consequence, we have a second intrinsic characterisation of the orthogonal geodesic separation (compare with \cite[Theorem~6, Section~3]{KM-1980} and \cite{Shapovalov-1981}; note that in~\cite[Theorem~6(4)]{KM-1980}  turns out to be redundant):

\begin{Theorem}\label{t:3.5}
The geodesic Hamilton--Jacobi equation is separable in orthogonal coordinates if and only if the Riemannian manifold admits a KS-space, i.e., a~$n$-dimensional linear space ${\cal K}$ of Killing tensors commuting as linear operators and in involution.
\end{Theorem}

In the applications, one of these last two theorems can be used according to the convenience. By applying Theorem \ref{t:3.2} we have the advantage of dealing with a single KT, but dif\/f\/iculties may arise in testing if it has simple eigenvalues and normal eigenvectors. Nevertheless, for solving this problem we can use the following two ef\/fective criteria:

\begin{Theorem}\label{t:3.6}
A $(1,1)$ tensor $\bs K$ has distinct eigenvalues if and only if
\begin{gather*}
D\doteq \vmatrix
n & S_1 & \ldots & S_{n-1} \cr
S_1 & S_2 & \ldots & S_n \cr
\vdots & \vdots & \cdots & \vdots\cr
S_{n-1} & S_n & \ldots & S_{2n-2}
\endvmatrix\neq 0, \qquad S_p\doteq \operatorname{tr}\big(\bs K^p\big).
\end{gather*}
 \end{Theorem}

Here $\bs K^p$ is the power $p=0,1,2,\ldots$ of the linear mapping $\bs K$. This theorem is a consequence of a classical theorem of Sylvester about the {\it discriminant} $D$ of an algebraic equation, here applied to the characteristic equation of $\bs K$.

\begin{Theorem}\label{t:3.7}
A symmetric tensor $\bs K$ with simple eigenvalues has normal eigenvectors if and only if
\begin{gather*}
H_{ab}^h   K^a_i K^b_j+2  H_{a[i}^b   K^a_{j]}   K^h_b+ H^a_{ij}  K^h_b K^b_a=0,
\end{gather*}
where $H$ is the Nijenhuis torsion of~$\bs K$,
\begin{gather*}
H_{ij}^h(\bs K)\doteq 2  K^a_{[i} \partial_{|a|}K_{j]}^h-2  K^h_a \partial_{[i}K^a_{j]}.
\end{gather*}
\end{Theorem}

This is a special case of a more general theorem due to Haantjes \cite{Haantjes-1955} (see also \cite[p.~248]{Schouten-1954}). As will be seen in Section~\ref{section7}, it is interesting the case of a {\it torsionless} tensor: $\bs H(\bs K)=0$. We will apply the following

\begin{Theorem}[\cite{Nijenhuis-1951}]\label{t:3.8}
A symmetric tensor $\bs K$ with simple eigenvalues $\rho^i$ is torsionless if and only if it has normal eigenvectors
$\bs X_i$ such that
\begin{gather*}
\bs X_i \rho^j=0, \qquad i\neq j.
\end{gather*}
\end{Theorem}

This means that each eigenvalue $\rho^j$ is constant on the hypersurfaces orthogonal to the corresponding eigenvector~$\bs X_j$. It is worthwhile to observe that

\begin{Theorem}
A torsionless KT with simple eigenvalues is necessarily a constant KT on a flat Riemannian manifold.
\end{Theorem}

\begin{proof}
If $\bs H(\bs K)=0$, then equations \eqref{(2.2)} imply
\begin{gather*}
0=\partial_i\rho^j=\big(\rho^i-\rho^j\big) \partial_i\ln g^{jj},\qquad i\neq j, \qquad \partial_i\rho^i=0.
\end{gather*}
This shows that $\rho^i={\rm const}$ and that $\partial_ig^{jj}=0$ for $i\neq j$. This last condition means that~$g^{jj}$ is a function of $q^j$ only. In this case, up to a change of scale of the coordinates, we can consider $g^{jj}={\rm const}$.
\end{proof}

This is a case considered in \cite{Bruce-McL-Smirnov-2001}. Going back to Theorem \ref{t:3.5}, it is interesting to make a comparison with the intrinsic characterisations of the geodesic orthogonal separation due to Eisenhart \cite{Eisenhart-1934,Eisenhart-1949} and Woodhouse \cite{Woodhouse-1975}. In the {\it Eisenhart theorem} \cite[p.~289]{Eisenhart-1934} the necessary and suf\/f\/icient conditions for the orthogonal separation are: (i)~the existence of $n-1$ independent Killing tensors $\bs K_1,\ldots, \bs K_{n-1}$ with normal common eigenvectors and such that (ii)~for {\it each} of these tensors the eigenvalues are simple and (iii)~for any pair $(i,\alpha)$ of f\/ixed indices ($\alpha =2,\ldots, n$, $i=1,\ldots, n$) the square matrices $\|\rho_i^\alpha-\rho_j^\alpha\|$ (with $j\neq i$) are regular. In the Eisenhart notation, $\rho_i^\alpha$ are the eigenvalues of $\bs K_\alpha$. Condition (i) should be replaced by the existence of $n-1$ Killing tensors such that $\bs G,\bs K_1,\ldots, \bs K_{n-1}$ are independent. Then Theorem~\ref{t:3.5} shows that conditions~(ii) and~(iii) are redundant. In~\cite[Theorem~4.2]{Woodhouse-1975} the $n-1$ KT's are assumed to be in involution and with common {\it closed} eigenforms. Theorem~\ref{t:3.3} shows that the requirement {\it closed} is redundant, since it is equivalent to the normality of the eigenvectors.

\section{The orthogonal separation of a natural Hamiltonian}\label{section4}

With each symmetric 2-tensor $\bs K$ and a function $U$ on $Q$ we associate the function $F=\tfrac 12  P_{\bs K}+U$ on $T^*Q$. We observe that {\it $F$ is a first integral of $H=G+V$, $\{H,F\}=0$, if and only if}
\begin{gather*}
\{G,P_{\bs K}\}=0, \qquad dU=\bs K\,dV.
\end{gather*}
The f\/irst equation means that $\bs K$ is a Killing tensor. If it is a ChKT, then, in any orthogonal coordinate system determined by its eigenvectors, the second equation is equivalent to $\partial_iU=\rho^i\partial_iV$.
Due to the fundamental equations~\eqref{(2.2)}, the integrability conditions of these equations assume the form
$\partial_j\partial_iU-\partial_i\partial_jU=(\rho^i-\rho^j) S_{ij}(V)=0$. This proves

\begin{Theorem}\label{t:4.1}
If $\bs K$ is a symmetric $2$-tensor with simple eigenvalues and normal eigenvectors, then~$F$ is a first integral of $G+V$ if and only if~$\bs K$ is a Killing tensor and $S_{ij}(V)=0$ in any orthogonal system of coordinates generated by the eigenvectors.
\end{Theorem}

Since the existence of coordinates is a local matter, condition $d U=\bs K\,dV$ can be replaced by $d(\bs K\,dV)=0$. Thus, by recalling Theorem~\ref{t:3.2}  and the remarks at the end of Section~\ref{section1},  we f\/ind that

\begin{Theorem}[\cite{Benenti-1993}]\label{t:4.2}
The Hamilton--Jacobi equation $G+V=E$ is separable in orthogonal coordinates if and only if there exists a Killing $2$-tensor $\bs K$ with simple eigenvalues and normal eigenvectors such that
\begin{gather}
d(\bs K\,dV)=0.
\label{(4.1)}
\end{gather}
\end{Theorem}

This equation has been called {\it characteristic equation} of a {\it separable potential} $V$.
We observe that for $n=2$ any vector f\/ield is normal. Since it can be proved that on a two-dimensional manifold the separation always occurs in orthogonal coordinates \cite{LC-1904}, we get

\begin{Theorem}\label{t:4.3}\sloppy
On a two-dimensional Riemannian manifold the Hamilton--Jacobi equation \mbox{$G+V=E$} is separable if and only if there exists a $($non-trivial$)$ quadratic first integral.
\end{Theorem}

This is the extension to a two-dimensional manifold of the so-called Bertrand--Darboux--Whittaker theorem for the Euclidean plane $\mathbb E_2$ \cite{Ankiewicz-1983,Whittaker-1937}.

When written in Cartesian coordinates on a Euclidean $n$-space, equation~\eqref{(4.1)} gives rise to the so-called {\it Bertrand--Darboux} (BD) equations.
If we know the form of all characteristic tensors of a manifold, then equation \eqref{(4.1)}, written in any coordinate system (even not separable), provides an ef\/fective criterion for the separability of the Hamilton--Jacobi equation. This criterion have been applied for instance in the study of the super-separability of the inverse-square three-dimensional {\it Calogero system}~\cite{BCR-2000}. In~\cite{BCR-2000} you can f\/ind the intrinsic (i.e., ``boldface'') expressions of all the characteristic Killing 2-tensors in the Euclidean three-space $\mathbb E_3$, so that this criterion is ready to be used for any potential $V$. For a general~$\mathbb E_n$, the basic ChKT's and the corresponding BD equations have been determined in Marshall and Wojciechowski, and in~\cite{Benenti-1992b,Benenti-1993}. This analysis has been completed by Waksj\"o~\cite{Waksjo-2000}. In his thesis he presents an ef\/fective general criterion for the separability of a potential $V$ in the Euclidean $n$-space (see also~\cite{Waksjo-Rauch-Wojciechowski-2003}). Other separability criteria can be based on the analysis of the fundamental invariants of spaces of Killing tensors under the action of isometry groups and the method of moving frames (see \cite{DHMcLS-2004, McLenaghan-Smirnov-2002}, also for related references).

If in $\mathbb E_n$ we take
\begin{gather}
\bs K=\text{tr}(\bs L) \bs G-\bs L, \qquad
\bs L\doteq\bs A+ m \bs r\otimes\bs r + \bs w \odot \bs r,
\label{(4.2)}
\end{gather}
where $\bs A$ is symmetric and constant, $m\in\mathbb R$, $\bs w$ is a constant vector, and $\bs r$ is the vector representing the generic point,
then in Cartesian coordinates $(x^i)$ equation~\eqref{(4.1)} yields the Bertrand--Darboux equations~(3.25) in~\cite{Waksjo-2000}. The correspondence of notation is the following: $\bs A=(\gamma_{ij})$, $m=\alpha$, $\bs w =(2 \beta_i)$).

For $\bs w=0$ we get the BD equations for the separation in elliptic coordinates centered at the origin. For $m=0$ we get the BD equations for the separation in parabolic coordinates centered at the point~$P$, where $\bs L_P(\bs w)$ (the existence of such a point is proved in~\cite{Benenti-1992b}). For $\bs w=0$ and $m=0$ we have the separation in Cartesian coordinates. In the remaining case $\bs w\ne 0$ and $m\ne 0$, we have the separation in elliptic coordinates centered at the point $\bs c= -\tfrac 1{2m}\bs w$ (see below).

It must be emphasised that in the characteristic equation \eqref{(4.1)} for a separable potential $V$ the eigenvalues of $\bs K$ must be simple (outside a {\it singular set}). For the present case we have

\begin{Theorem}\label{t:4.4}
The tensor $\bs K$ defined in \eqref{(4.2)} has simple eigenvalues if and only if for $m=0$ the eigenvalues of $\bs A$ are simple and for $m\neq 0$ the eigenvalues of $\bs A-\tfrac 1{4m} \bs w\otimes\bs w$
are simple.
\end{Theorem}

This is (a slightly modif\/ied version of) a theorem of Waksj\"o \cite[p.~45]{Waksjo-2000}. We give here an alternative proof.

\begin{proof}
The tensor $\bs K$ has simple eigenvalues if and only if $\bs L$ has simple eigenvalues.
For $m=0$ we have $\bs L=\bs A+\bs w\odot \bs r$; we are in the case of the parabolic web~\cite{Benenti-1992b} and~$\bs L$ has simple eigenvalues if and only if $\bs A$ has simple (constant) eigenvalues.
For $m\neq 0$, let us change the origin by considering the transformation $\bs r=\bs c+\bs r'$,
where $\bs c$ is a constant vector. We obtain
\begin{gather*}
\bs L=\bs A+m  (\bs r'\otimes \bs r'+2 \bs r'\odot\bs c+\bs c\otimes\bs c)+
\bs w\odot \bs c +\bs w\odot\bs r'.
\end{gather*}
If we take $\bs c= -\tfrac 1{2m} \bs w$, then
\begin{gather*}
\bs L=\bs A-\tfrac 1{4m}   \bs w\otimes\bs w + m  \bs r'\otimes \bs r'=\bs A'+m  \bs r'\otimes \bs r'.
\end{gather*}
We known that for a symmetric tensor of the kind $\bs A'+\bs r'\otimes\bs r'$ the eigenvalues $(u^i)$
are the roots of equation
\begin{gather*}
\sum_i\dfrac{x_i^2}{u-a^i}=\frac 1m,
\end{gather*}
where $(x_i)$ are Cartesian coordinates and $(a^i)$ are the constant eigenvalues of $\bs A'$.
This equation is equivalent to
\begin{gather}
m \sum_i x_i^2 \prod_{k\neq i} \big(u-a^k\big)- \prod_j\big(u-a^j\big)=0.
\label{(4.3)}
\end{gather}
If $(a^i)$ are simple, then also $(u^i)$ are simple, since $a^1< u^1 < a^2 < u^2 < \ldots < u^{n-1} < a^n < u^n$. If~$(a^i)$ are not all simple, for instance $a^1=a^2$, then equation~\eqref{(4.3)} has a double root $u=a^1$.
\end{proof}

The tensors $\bs K$ and $\bs L$ have a mechanical meaning: they were introduced in~\cite{Benenti-1992b} as the {\it inertia tensor} and {\it planar inertia tensor} of a set of massive points (including, this is important, negative masses) in $\mathbb E_n$. The parameter $m$ is just the total mass (it may be $0$). Indeed, it is a remarkable fact that {\it an inertia tensor} is a~KT. This interpretation is of help in the problem of f\/inding the intrinsic expressions of all the ChKT's of $\mathbb E_n$ (see also~\cite{Marshall-Wojciechowski-1988}).

\begin{Notation}  A matrix of the kind $\bs A+\bs r\otimes\bs r$ has been used by Moser~\cite{Moser-1981} for constructing a~Lax pair for the geodesic f\/low of an asymmetric ellipsoid. For this reason it was denoted by~$\bs L$ in~\cite{Benenti-1992b} (indeed, in analogy with the Lax method, starting from~$\bs L$ we can construct  a complete system of f\/irst integrals in involution through a pure algebraic process). There were other two reasons which suggested this notation: (i)~$\bs L$ is a letter adjacent to~$\bs K$, and this is appropriate because a~tensor~$\bs L$ generates a tensor~$\bs K$ according to~\eqref{(4.2)}; (ii)~$\bs L$ stands for {\it Levi-Civita}, and indeed the orthogonal metric associated with $\bs L$ was f\/irstly introduced by Levi-Civita~\cite{LC-1896} (see the end of Section~\ref{section7}).
\end{Notation}

\section{First integrals associated with the orthogonal separation}\label{section5}

Going back to the characteristic equation \eqref{(4.1)} we recall that at the beginning of Section~\ref{section3} we observed that a characteristic Killing tensor, like that appearing in equation \eqref{(4.1)}, generates a~KS-space~$\cal K$. It is a~remarkable fact that

\begin{Theorem}\label{t:5.1}
If $(\bs K_a)=(\bs K_0, \bs K_2, \ldots, \bs K_{n-1})$ is a basis of $\cal K$  $(\bs G$ and $\bs K$ may belong to this basis$)$ then locally there exist functions $V_a$ such that
\begin{gather}
H_a\doteq \tfrac 12 P_{\bs K_a}+V_a
\label{(5.1)}
\end{gather}
are independent first integrals in involution.
\end{Theorem}

A proof can be found in~\cite{Benenti-1993}. Indeed, it can be shown that equation \eqref{(4.1)} implies $d(\bs K_a\,dV)=0$ for each index $a$. This implies that a ChKT $\bs K$ and a separable potential $V$ generates a $n$-dimensional space ${\cal H}(\bs K,V)$ of f\/irst integrals in involution,
\begin{gather*}
H_{\bs K}=\tfrac 12 P_{\bs K}+V_{\bs K}, \qquad \bs K\in \cal K,
%\label{(5.2)}
\end{gather*}
and that the associated potentials $V_{\bs K}$ can be determined by integrating the closed 1-forms $\bs K_a\,dV$. We observe that there are separable systems (an example is the three-body Calogero system, see \cite{BCR-2000}) in which this integration can be avoided and replaced by an algebraic process.

We remark that, if we know a basis of a KS-space, then for testing if a potential $V$ is separable it is suf\/f\/icient to verify that equation \eqref{(4.1)} is satisf\/ied for an element of this space with simple eigenvalues. When the answer is af\/f\/irmative, then a complete set of integrals in involution can be determined by integrating the closed 1-forms $\bs K_a\,dV$.

It is well known that the orthogonal (as well as the non-orthogonal separation) is related to St\"ackel matrices. A {\it St\"ackel matrix} in the $n$ variables $(q^i)$ is a regular $n\times n$ matrix $\bs S=\big[\varphi_i^{(a)}\big]$ of functions $\varphi_i^{(a)}$ depending on the variable $q^i$ corresponding to the lower index only. We denote by~$\big[\varphi_{(a)}^i\big]$ the inverse matrix. The original {\it St\"ackel theorem}
asserts that {\it an orthogonal coordinate system $(q^i)$ is separable if and only if there exists a St\"ackel matrix such that $g^{ii}=\varphi_{(0)}^i$. In this case},  (i) {\it the diagonalised tensors $K_a^{ii}=\varphi_{(a)}^i$ are the basis of a St\"ackel space}, (ii)~{\it all separable potentials $V$ have the form $V=\phi_i(q^i) \varphi_{(0)}^i$,
where $\phi_i$ is a function of the corresponding coordinate~$q^i$ only};  (iii)~{\it a basis $H_a$  of the space of
f\/irst integrals in involution is given by~\eqref{(5.1)} with $V_a=\phi_i \varphi_{(a)}^i$}.\footnote{St\"ackel systems have been the object of several researches in recent years. A~generalisation of St\"ackel systems has been proposed in~\cite{Blaszak-Sergyeyev-2011}. Relations between f\/inite- (St\"ackel systems) and inf\/inite-dimensional (Harry--Dym hierarchies) integrable systems are investigated in~\cite{Marciniak-Blaszak-2010}. The separability of  a~Hamiltonian  in dif\/ferent St\"ackel systems is a standard procedure to determine  superintegrable systems (see \cite{Miller-Post-Winternitz-2013} and references therein). More links between St\"ackel systems and superintegrability are shown in~\cite{Tsiganov-2010}. An analysis of the Killing--St\"ackel spaces is the basis for a classif\/ication of the orthogonal separable coordinates  in~\cite{Rajaratnam-McLenaghan-2014}. (Bibliographical footnotes added by Giovanni Rastelli.)\label{pagenote1}}

\section{Conformal Killing tensors}\label{section6}

As remarked above, the existence of a KT $\bs K$ with simple eigenvalues and normal eigenvectors is a necessary and suf\/f\/icient condition for the existence of a KS-space $\cal K$, i.e., of a $n$-dimensional linear space of KT's with common normal eigenvectors (and, consequently, in involution). The following question arises: {\it is it possible to construct a basis of the space $\cal K$ by a coordinate independent algebraic procedure, starting from~$\bs K$}? Note that this problem can be solved (in principle) by integrating the linear dif\/ferential system \eqref{(2.2)}, if we know a separable coordinate system.

As illustrated in \cite{Benenti-1992b} the answer is af\/f\/irmative at least for special kinds of St\"ackel systems. In the next sections we will revisit this matter, by proposing new def\/initions and theorems. To this end, we need to recall some basic def\/initions and properties concerning conformal Killing tensors.

A {\it conformal Killing tensor} (CKT) on a Riemannian manifold $Q_n$ is a symmetric tensor $\bs L$ of order $l$ satisfying equation $\{P_{\bs L},P_{\bs G}\}=P_{\bs X} P_{\bs G}$, where $\bs X$ is a suitable symmetric tensor of order $l-1$. Since we are interested in CKT's of order two, we write this equation in the form
\begin{gather}
\{P_{\bs L},P_{\bs G}\}= -2 P_{\bs C} P_{\bs G},
\label{(6.1)}
\end{gather}
where $\bs C$ is a vector f\/ield which we call {\it associated with} $\bs L$ (also denoted by $\bs C(\bs L)$).

A CKT is said to be of {\it gradient-type} (GCKT) if $\bs C=\nabla f$. An example of GCKT is $f \bs G$. In this case we have $\bs C=\nabla f$. Indeed, $\{P_{f\bs G},P_{\bs G}\}=\{fP_{\bs G},P_{\bs G}\}=
P_{\bs G}\{f,P_{\bs G}\}= -2g^{ij}p_j\partial_if P_{\bs G}$.
A KT is obviously a GCKT with $\bs C=0$.

\begin{Theorem} \label{t:6.1}
A CKT $\bs L$ is of gradient-type with $\bs C=\nabla f$ if and only if
$\bs K=f \bs G-\bs L$ is a~Killing tensor.
\end{Theorem}

\begin{proof} $\{P_{\bs L}-fP_{\bs G}, P_{\bs G}\}=\{P_{\bs L},P_{\bs G}\}-P_{\bs G}\{f,P_{\bs G}\}= -2P_{\bs C}P_{\bs G}+2P_{\bs G}P_{\nabla f}$.
\end{proof}

A {\it conformal Killing tensor of trace-type} $\bs L$ is a CKT for which $\bs C=\nabla\operatorname{tr}(\bs L)$.
In this case, $\bs K= \operatorname{tr}(\bs L) \bs G-\bs L$ is a~KT.

By considering in \eqref{(6.1)} $P_{\bs G}=g^{ii}p_i^2$ and $P_{\bs L}=u^ig^{ii}p_i^2$, we f\/ind

\begin{Theorem}\label{t:6.2}
Assume that a symmetric $2$-tensor $\bs L$ is diagonalised in orthogonal coordinates, so that $L^{ii}=u^i g^{ii}$, $L^{ij}=0$, $i\neq j$, where~$u^i$ are the eigenvalues. Then~$\bs L$ is a CKT if and only if the following equations are satisfied,
\begin{gather}
\partial_iu^k=\big(u^i-u^k\big) \partial_i\ln g^{kk}+C_i, \qquad
C_i=\partial_iu^i.
\label{(6.2)}
\end{gather}
\end{Theorem}

\section{L-tensors}\label{section7}

Let us call {\it L{\rm-}tensor} a (i) conformal Killing tensor $\bs L$ with (ii) vanishing torsion and (iii) pointwise simple eigenvalues $(u^i)$. The reasons for introducing such an object will be explained in the next section. In the present section we examine the basic properties of an L-tensor. Due to the vanishing of the torsion, there is an equivalence classes of orthogonal coordinates $(q^i)$ in which this tensor is diagonalised and $\partial_iu^j=0$ for $i\neq j$. We say that these coordinates are {\it associated} with~$\bs L$. Since $\bs L$ is a CKT, due to Theorem~\ref{t:6.2} equations
\begin{gather}
\partial_iu^j=\big(u^i-u^j\big)\partial_i\ln g^{jj}+C_i=0, \qquad i\neq j, \qquad
C_i=\partial_iu^i
\label{(7.1)}
\end{gather}
hold. For $\bs C=0$ we f\/ind the equations \eqref{(2.2)} of a KT.

\begin{Remark}  In the above def\/inition no assumption is made about the independence of the eigenvalues as functions on~$Q$; some of them may be constant (a criterion for the independence of the eigenvalues is given in Theorem~\ref{t:9.2} below). This def\/inition has to be compared with those given in \cite{Blaszak-2003, Bolsinov-Matveev-2003, Ibort-Magri-Marmo-2000}, where $\bs L$ is assumed to be a torsionless CKT of trace-type (and called {\it Benenti tensor}), and in \cite{Crampin-2003a,Crampin-2003b}, where~$\bs L$ is assumed to be a torsionless CKT with independent (i.e., coordinate-forming) eigenvalues (and called {\it special conformal Killing tensor}, see Theorem~\ref{t:9.1} below). In all these papers the essential condition that~$\bs L$ has simple eigenvalues is missing (or understood).
\end{Remark}

In fact it can be proved that

\begin{Theorem}\label{t:7.1}
If an eigenvalue $u^i$ of a torsionless CKT of trace-type is not simple, then it is constant.
\end{Theorem}

The proof of this theorem (here omitted) requires the use of the Haantjes theorem for a tensor with non-simple eigenvalues.

\begin{Theorem}\label{t:7.2}\sloppy
Let $\bs L$ be an L-tensor with associated coordinates $(q^i)$. Then:
$(i)$~Each eigenva\-lue~$u^i$ depends on the associated coordinate $q^i$ only, $u^i=u^i(q^i)$.
$(ii)$~It is of trace-type, $\bs C=\nabla\operatorname{tr}(\bs L).$
$(iii)$~In associated coordinates the metric has the form
\begin{gather}
g^{kk}=\phi_k \prod_{i\neq k}\dfrac 1{|u^i-u^k|}, \qquad u^i=u^i\big(q^i\big), \qquad\phi_k=\phi_k\big(q^k\big)>0,
\label{(7.2)}
\end{gather}
or, after a rescaling
\begin{gather}
g^{kk}=\prod_{i\neq k}\dfrac 1{|u^i-u^k|}.
\label{(7.3)}
\end{gather}
In both cases,
\begin{gather*}
\partial_i\ln g^{kk}=\dfrac {\partial_i u^i}{u^k-u^i},\qquad  i\neq k.
\end{gather*}
$($We call  normal coordinates associated with $\bs L$ the orthogonal coordinates for which equa\-tions~\eqref{(7.3)} hold$)$. $(iv)$~The associated coordinates are separable.
$(v)$~$\bs L$ commutes with the Ricci ten\-sor~$\bs R$, $[\! [\bs L,\bs R]\! ]=0$, i.e.,  the Robertson condition is satisfied: in the associated coordinates, $R_{ij}=0$ for $i\neq j$.
\end{Theorem}

These properties are derived from \cite{Benenti-1992b,Benenti-1993}. They follow from equations~\eqref{(7.1)} and from the fundamental properties of the elementary symmetric polynomials (see the next section).
From~\eqref{(7.2)} it follows that the {\it contracted Christoffel symbols} $\Gamma_i=g^{hj}\Gamma_{hj,i}$ take the simple form
$\Gamma_i=-\tfrac 12 \phi'_k\phi_k$. Thus, in normal coordinates $\Gamma_i=0$. Since in separable orthogonal coordinates
$R_{ij}=\tfrac 32 \partial_i\Gamma_j$, for $i\neq j$, the Robertson condition~(v) is proved. Item (ii) (which follows from \eqref{(7.1)}: $C_i=\partial_iu^i$ and $\partial_iu^j=0$ for $i\neq j$ implies $C_i=\partial_i\sum_ju^j$) shows that in~\cite[Proposition~1, Section~1]{Crampin-2003a} the assumption that~$(u^i)$ are functionally independent eigenvalues is redundant.

It is a remarkable fact that the metric~\eqref{(7.2)} associated with $\bs L$ is that of the {\it correspon\-ding geodesics} found by Levi-Civita~\cite{LC-1896}: {\it there exists a metric $\overline{\bs G}=(\overline g^{ij})$ having the same $($unparametrized$)$ geodesics of a given metric $\bs G=(g^{ij})$ if and only if there are orthogonal coordinates in which the metric $\bs G$ assumes the form~\eqref{(7.2)}}. It must be pointed out that this theorem holds under the assumption that the tensor $\overline {\bs G}$ has simple eigenvalues with respect to  $\bs G$. This matter has been recently analysed by Bolsinov and Matveev~\cite{Bolsinov-Matveev-2003} and by Crampin~\cite{Crampin-2003b}. The metric~\eqref{(7.2)} is a special case of the orthogonal separable metric determined by Eisenhart~\cite[Appendix~13]{Eisenhart-1949}
and characterised by the condition~$R_{jiik}=0$ for $i,j,k\neq$.

\section{L-systems}\label{section8}

Let $\sigma_a(\underline u)$ denote the elementary symmetric polynomial of degree $a$ of the $n$ variables $\underline u=(u^i)$. Let $\sigma_a^i$ and $\sigma_a^{ij}$ be the functions obtained from $\sigma_a$ by setting $u^i=0$ and $u^j=0$. Let us set
\begin{gather}
\sigma_0=\sigma_0^i=\sigma_0^{ij}=1, \qquad \sigma_{-1}=\sigma_{-1}^i=\sigma_{-1}^{ij}=0,\qquad\sigma^{ij}_n=\sigma_{n-1}^{ij}=0.
\label{(8.1)}
\end{gather}
Then the following equations are satisf\/ied \cite[Section~2]{Benenti-1992b}:
\begin{gather}
\sigma_a=\sigma_a^i+u^i \sigma_{a-1}^i, \qquad \sigma_a^j=\sigma_a^{ij}+u^i \sigma_{a-1}^{ij}, \qquad \sigma_a^i-\sigma_a^j=\big(u^j-u^i\big) \sigma^{ij}_{a-1}.
\label{(8.2)}
\\
 \sum_i u^i \sigma_{a-1}^i=a \sigma_a, \qquad
\det \big[\sigma_a^i\big]= \prod_{j>i}\big(u^i-u^j\big).
\label{(8.3)}
\\
\dfrac{\partial \sigma_a}{\partial u^i}=\sigma^i_{a-1},\qquad
\dfrac{\partial \sigma_a^j}{\partial u^i}=\sigma_{a-1}^{ij},\qquad
\big(u^j-u^i\big)\dfrac{\partial \sigma_a^j}{\partial u^i}=\sigma_a^i-\sigma_a^j.
\label{(8.4)}
\end{gather}
If $u^i\neq u^j$ for $i\neq j$, then
\begin{gather*}
\dfrac{\partial \sigma_a^j}{\partial u^i}=\dfrac{\sigma_a^i-\sigma_a^j}{u^j-u^i},
\qquad
\det \big[\sigma_a^i\big]\neq 0, \qquad a=0,1,\ldots, n-1,\qquad i=1,\ldots,n.
%\label{(8.5)}
\end{gather*}
It follows that for any coordinate system $(q^i)$,
\begin{gather}
\partial_i\sigma_a^j=\partial_iu^h\;\dfrac{\partial\sigma_a^j}{\partial u^h}=\sum_{h\neq j}\partial_iu^h\;\sigma_{a-1}^{hj}=
\sum_{h\neq j}\partial_iu^h\;\dfrac{\sigma_a^h-\sigma_a^j}{u^j-u^h}.
\label{(8.6)}
 \end{gather}

\begin{Theorem}\label{t:8.1}
Let $\bs L$ be a symmetric $2$-tensor with eigenvalues $(u^i)$. The tensors $(\bs K_a)=(\bs K_0,\bs K_1,\ldots,\bs K_{n-1})$ defined by
\begin{gather}
\bs K_0=\bs G,\qquad
\bs K_a=\tfrac 1a \operatorname{tr}(\bs K_{a-1}\bs L)\bs G-\bs K_{a-1}\bs L, \qquad { {a>0}},
\label{(8.7)}
\end{gather}
or by
\begin{gather}
\bs K_a=\sum_{k=0}^a (-1)^k\sigma_{a-k}\bs L^k, \qquad \bs K_a=\sigma_a \bs G-\bs K_{a-1}\bs L, \qquad \bs K_{-1}=\bs 0,
\label{(8.8)}
\end{gather}
form a basis of a KS-space if and only if $\bs L$ is an L-tensor.
\end{Theorem}

\begin{proof}
(i) Assume that $(\bs K_a)$ def\/ined by \eqref{(8.7)} is a basis of a KS-space. Then they are linearly independent and there are orthogonal (separable) coordinates in which these tensors as well as~$\bs L$ are diagonalised and equations
\begin{gather}
\partial_i\rho_a^j=\big(\rho_a^i-\rho_a^j\big) \partial_i\ln g^{jj},\qquad \det\big[\rho_a^i\big]\neq 0,
\label{(8.9)}
\end{gather}
hold, being $\rho_a^i$ the eigenvalues of~$\bs K_a$. Due to~\eqref{(8.7)}, these eigenvalues fulf\/ill the recurrence relation
$\rho_a^i=\tfrac 1a \sum_k\rho_{a-1}^ku^k-\rho_{a-1}^i u^i$. On the other hand, from the f\/irst equations~\eqref{(8.3)} and~\eqref{(8.2)} we get $\tfrac 1a \sum_k\sigma_{a-1}^k u^k-\sigma_{a-1}^i u^i=\sigma_a-\sigma_{a-1}^i u^i=\sigma_a^i$. This shows that
\begin{gather*}
\rho_a^i=\sigma_a^i(\underline u).
\end{gather*}
It follows that: (I)~Due to the second equations \eqref{(8.3)} and~\eqref{(8.9)}, the eigenvalues $u^i$ are simple. (II)~Due to the f\/irst equation \eqref{(8.2)}, the def\/inition \eqref{(8.7)} implies the alternative def\/initions~\eqref{(8.8)}. (III)~Due to the f\/irst equation \eqref{(8.9)},  $\partial_i\rho_a^i=0$ thus, $0=\sum\limits_{h\neq i}\partial_iu^h \sigma_{a-1}^{hi}$.
Let us consider the case $i=1$. We get the linear homogeneous system of $n-1$ equations
\begin{gather}
\sum_{h>1}\partial_1u^h \sigma_{a-1}^{h1}=0, \qquad a=1,\ldots,n-1,
\label{(8.10)}
\end{gather}
in the $n-1$ unknown functions $\partial_1u^h$, with $h=2,\ldots,n$.
We can put $\sigma_{a-1}^{h1}=\tilde\sigma_b^h$, where $\tilde\sigma_b^h$, $b=0,\ldots, n-2$, are the symmetric polynomials in the $n-1$ variables $(u^2,\ldots,u^n)$. In analogy with the third equation~\eqref{(8.4)} we have $\det[\tilde\sigma_b^h]=\prod\limits_{j>i>1}(u^i-u^j)$, thus $\det[\tilde\sigma_b^h]=\det[\sigma_{a-1}^{h1}]\neq 0$. It follows from~\eqref{(8.10)} that $\partial_1u^h=0$ for all $h>1$. In a similar way we prove that
$\partial_iu^h=0$  for all $h\neq i$. This shows that $\bs H(\bs L)=0$.
Finally, from $\bs K_1=\text{tr}(\bs K_1\bs L)-\bs L$ we get
\begin{gather}
\bs L\doteq \frac 1{n-1}\operatorname{tr}(\bs K_1) \bs G-\bs K_1.
\label{(8.11)}
\end{gather}
Being $\bs K_1$ a KT, $\bs L$ is a CKT.
(ii) Conversely, assume that $\bs L$ is an L-tensor. In coordinates associated with $\bs L$ we have
$\partial_iu^h=0$ for all $h\neq i$, and moreover, $0=(u^i-u^j)\partial_i\ln g^{jj}+\partial_iu^i$.
By~\eqref{(8.6)} we get
\begin{gather*}
\partial_i\sigma_a^j=\partial_iu^i \dfrac{\sigma_a^i-\sigma_a^j}{u^j-u^i}=\big(\sigma_a^i-\sigma^j_a\big) \partial_i\ln g^{jj}.
\end{gather*}
Being $\rho_a^i=\sigma_a^i$, this shows that the tensors $\bs K_a$ are Killing tensors. They are pointwise independent due to~\eqref{(8.3)}.
\end{proof}

We call {\it L-system} any separable orthogonal system whose KS-space is generated by an L-tensor according to Theorem~\ref{t:8.1}. We call {\it L{\rm -}web} any orthogonal web corresponding to an L-system\footnote{L-systems have been considered in~\cite{Blaszak-Marciniak-2006, Blaszak-Sergyeyev-2005}. See also the  extensive classif\/ication~\cite{Benenti-2005} of special symmetric two-tensors (including L-tensors) with application to the theory of the equivalent dynamical systems and to the theory of cofactor and bi-cofactor systems.\label{pagenote2}}.

\begin{Remarks}\quad
\begin{enumerate}\itemsep=0pt
\item[(i)] All the Killing tensors $\bs K_a$ ($a\neq 0$) have simple eigenvalues. The sequence~\eqref{(8.7)} was suggested by the analysis \cite{Benenti-1992b} of the planar inertia tensor $\bs L$ of an asymmetric massive body in the Euclidean $n$-space. The tensor $\bs K_1$ is the corresponding inertia tensor.
Formula~\eqref{(8.8)} shows that $\bs K_a=0$ for $a>n-1$, since for $a=n$ the right-hand side vanishes due to the Hamilton--Cayley theorem. Formula~\eqref{(8.7)} is more ef\/fective than~\eqref{(8.8)} since it does dot require the knowledge of the eigenvalues of~$\bs L$.

\item[(ii)] Within  a completely dif\/ferent context, a sequence like~\eqref{(8.7)} is considered by Schouten  \cite[p.~30]{Schouten-1954} generated by a~matrix~$\bs P$, as an ef\/fective tool for computing the eigenvectors of a~matrix $\bs P$ without solving systems of linear equations (when the eigenvalues are known and simple). Such a method was f\/irstly introduced by Fettis~\cite{Fettis-1950} and Souriau~\cite{Souriau-1950}.

\item[(iii)] We observe with Schouten \cite{Schouten-1954} that the (1,1) tensor $\bs Q(x)=\operatorname{cof}(\bs L-x\bs G)$ is polynomial of degree $n-1$ in $x$, whose coef\/f\/icients, up to the sign, are the tensors $\bs K_a$ def\/ined in~\eqref{(8.8)}. We recall that the cofactor $\tilde{\bs A}=\operatorname{cof}(\bs A)$ of $\bs A$ is def\/ined by $\bs A\tilde{\bs A}=\tilde{\bs A}\bs A=\det(\bs A)\bs G$. Hence, the St\"ackel systems of the kind considered in the last two theorems are just the so-called {\it cofactor systems} (cf.~\cite{Blaszak-Ma-2003}).
In fact, they should be called {\it Levi-Civita systems} (so that ``L-systems'' is a good notation) since the separable metric~\eqref{(7.3)} (but not the tensor $\bs L$) appears for the f\/irst time in~\cite{LC-1896}, where it is also shown that for such a metric the function
\begin{gather*}
F(\underline q,\underline{\dot q},c)\doteq  \prod_{j\neq i} \big(u^j+c\big) \big|u^j-u^i\big| \big(\dot q^i\big)^2
\end{gather*}
is a f\/irst integral of the geodesics for all values of the parameter $c$. Since $F$ is a~polynomial in~$c$ of degree~$n-1$, its coef\/f\/icients gives rise to~$n$ f\/irst integrals. These f\/irst integrals coincide, up to the sign and after the Legendre transformation, with the f\/irst integ\-rals~$P(\bs K_a)$.

\item[(iv)] Due to item~(v) of Theorem~\ref{t:7.2}, a necessary condition for a~St\"ackel system to be an L-system is the Robertson condition: $[\! [\bs K,\bs R]\! ]=0$ for a characteristic tensor $\bs K$ (thus, for all elements of the KS-space).
 \end{enumerate}
 \end{Remarks}

     A criterion for testing if a St\"ackel system is an L-system is the following.

\begin{Theorem}
A St\"ackel system is an L-system if and only if in the corresponding KS-space there exists a characteristic tensor~$\bs K_1$ such that the tensor~$\bs L$ defined by~\eqref{(8.11)} is torsionless.
\end{Theorem}

\begin{proof}
(i) Assume that there exists a ChKT $\bs K_1$ such that the tensor $\bs L$ def\/ined by \eqref{(8.11)} is torsionless. Any tensor of the kind $f\bs G+\bs K$ is a CKT if $\bs K$ is a KT. Since $\bs K_1$ has simple eigenvalues, also $\bs L$ has simple eigenvalues. Then $\bs L$ is an L-tensor. The tensors $\bs K_a$ constructed by applying~\eqref{(8.7)} form a KS-space $\cal K_*$ which has $\bs K_1$ in common with the original KS-space~$\cal K$.  Thus, $\cal K_*=\cal K$ since two KS-spaces with a ChKT in common coincide. (ii) The converse is obvious.
\end{proof}

About the {\it uniqueness} of an L-tensor generating an L-system, it can be proved that

\begin{Theorem}
Two L-tensors $\bs L$ and $\tilde{\bs L}$ generates the same L-system if and only if $\tilde{\bs L}=a\bs L+b\bs G$, $a,b\in\mathbb R$, $a\neq 0$.
\end{Theorem}

In accordance with the remarks at the end of Section~\ref{section5}, a basis
$(\bs K_a)$ of a KS-space corresponds to the inverse of a St\"ackel matrix
$\bs S^{-1}=\big[\varphi_{(a)}^i\big]$, $\varphi_{(a)}^i=K_a^{ii}$.
In the sequence~\eqref{(8.7)} we have \mbox{$\bs K_0=\bs G$}, then $\varphi_{(0)}^i=g^{ii}$ and $\varphi_{(a)}^i=K_a^{ii}=\sigma_a^ig^{ii}$.
By using the formulas concerning the elementary symmetric function at the beginning of Section~\ref{section8} it can be proved that

\begin{Theorem}
For an L-system the St\"ackel matrix associated with the basis~\eqref{(8.7)}, $\bs S=\big[\varphi_i^{(a)} \big]$, assuming $u^1<u^2<\cdots<u^n$, is of  alternating Vandermonde type,
\begin{gather}
\varphi_i^{(a)}=(-1)^{n-a+i}  (u^i)^{n-a-1}, \qquad
a=0,\ldots,n-1,\qquad
i=1,\ldots,n\quad \hbox{\rm index of row.}
\label{(8.12)}
\end{gather}
\end{Theorem}

For $n=3$ we have ($a$ index of row)
\begin{gather*}
\big[\varphi_{(a)}^i\big]=\bmatrix
g^{11} & g^{22} & g^{33}
\\[3pt]
(u^2+u^3)g^{11} & (u^3+u^1)g^{22} & (u^1+u^2)g^{33}
\\[3pt]
u^2u^3g^{11} & u^3u^1g^{22} & u^1u^2g^{33}
\endbmatrix,\\
g^{11}=\dfrac{1}{(u^1-u^2)(u^1-u^3)},
\qquad
g^{22}=\dfrac{1}{(u^2-u^3)(u^1-u^2)},
\qquad
g^{33}=\dfrac{1}{(u^3-u^2)(u^3-u^1)}.
\end{gather*}
In accordance with Theorem~\ref{t:8.1}, the inverse matrix is
\begin{gather*}
\bs S=[\varphi_i^{(a)}]=\bmatrix
(u^1)^2 & -u^1 & 1
\\[6pt]
- (u^2)^2 & u^2 & -1
\\[6pt]
(u^3)^2 & -u^3 & 1 \cr\endbmatrix.
\end{gather*}

\begin{Remark}  If we multiply each row $\varphi_i^{(a)}$ of a St\"ackel matrix by a function $f_i(q^i)\neq 0$, then we get a new St\"ackel matrix $\tilde\varphi_i^{(a)}=f_i\varphi_i^{(a)}$ whose inverse $\tilde\varphi^i_{(a)}$ def\/ines a new basis of the  same KS-space. By multiplying the lines of~\eqref{(8.12)} by $\pm 1$ in a suitable way and by changing $\bs L$ in $-\bs L$, we can get a St\"ackel matrix which is of the Vandermonde type (see the case $n=3$, for instance).
\end{Remark}

\section{The functional independence of the eigenvalues of an L-tensor}\label{section9}

The results of Section~\ref{section8} hold without any assumption on the functional independence of the eigen\-va\-lues~$(u^i)$ of $\bs L$. Some of them may be constant. The only essential assumption is that they are pairwise and pointwise distinct. The following theorems deal with this matter.

\begin{Theorem}\label{t:9.1}
If a torsionless CKT has functionally independent eigenvalues $(u^i)$, then it is an L-tensor $($its eigenvalues are pointwise simple$)$.
\end{Theorem}

\begin{proof}
In this case the eigenvalues def\/ine local orthogonal coordinates $q^i=u^i$ in which the tensor is diagonalised. We can apply Theorem~\ref{t:6.2}. Equations \eqref{(6.2)} become $\delta_i^k=(u^i-u^k)\partial_i\ln g^{kk}+1$. For $i\neq k$ we have $u^i\neq u^k$.
\end{proof}

Note that in this case $\bs L$ is a special type of L-tensor (the type considered by Crampin \cite{Crampin-2003a,Crampin-2003b}).

\begin{Theorem}\label{t:9.2}
Let $\bs L$ be an L-tensor. $(i)$  The eigenvalues $(u^i)$ of $\bs L$ are independent functions $($i.e., they def\/ine locally an orthogonal coordinate system$)$ if and only if~$\bs L$ is  not invariant with respect to  a Killing vector $\bs X$. $(ii)$ If there exists a Killing vector~$\bs X$ such that $[\bs X,\bs L]=0$, then~$\bs X$ is a linear combination $($with constant coefficients$)$ of Killing vectors in involution which are eigenvectors of~$\bs L$.
\end{Theorem}

In case (i) we have no {\it symmetry} of the separable web generated by the eigenvectors of~$\bs L$. This means that all the structures associated with~$\bs L$ are not invariant with respect to  groups of isometries.

\begin{proof}
The eigenvalues of a (1,1)-tensor~$\bs L$ are functionally independent, and thus they def\/ine a~coordinate system, if and only if $\det[\partial_iu^j]\neq 0$. For a~torsionless tensor with simple eigenvalues we have, with respect to  the associated coordinates, $\partial_iu^j=0$ for $i\neq j$. Hence, the eigenvalues are independent functions if and only $\partial_iu^i\neq 0$ for all indices. (i)~If an eigenvalue $u^i$ of $\bs L$ is constant, $\partial_iu^i=0$, then from the expression~\eqref{(7.3)} of the metric we see that the corresponding (normalised) coordinate $q^i$ is ignorable. This means that $\partial/\partial q^i$ is a Killing vector which leaves invariant $\bs L$. Conversely, let us assume that there is a Killing vector $\bs X$ such that $[\bs X,\bs L]=0$. These two assumptions on~$\bs X$ are equivalent to $\{P_{\bs X}, P_{\bs G}\}=0$ and $\{P_{\bs X}, P_{\bs L}\}=0$. In an orthogonal coordinate system $(q^i)$ associated with $\bs L$, these two equations read $\{X^ip_i, g^{jj}p_j^2\}=0$ and $\{X^ip_i, u^jg^{jj}p_j^2\}=0$, respectively. Since they are algebraic equations in $\underline p$, to be satisf\/ied for any value of these variables,they are equivalent to
\begin{gather*}
\sum_i X^i\partial_ig^{kk}-2 g^{kk} \partial_kX^k=0,
\qquad
 g^{jj} \partial_jX^k+g^{kk} \partial_kX^j, \qquad j\neq k,
\\
 \sum_i X^i\partial_i\big(u^k g^{kk}\big)-2 u^k g^{kk} \partial_kX^k=0,
\qquad u^j g^{jj} \partial_jX^k+u^k g^{kk} \partial_kX^j, \qquad j\neq k.
\end{gather*}
The f\/irst two equations characterise the Killing vectors in orthogonal coordinates. The second and the fourth equations imply $(u^j-u^k) g^{jj} \partial_jX^k=0$ for $j\neq k$. Since $u^j\neq u^k$ we conclude that $\partial_jX^k=0$ for $j\neq k$, which means that $X^i=X^i(q^i)$. Since $\partial_iu^k=0$ for $i\neq k$, from the third equation it follows that $X^k \partial_ku^k g^{kk}+u^k \sum_iX^i \partial_ig^{kk}-2 u^k g^{kk} \partial_kX^k=0$. Due to the f\/irst equation~\eqref{(8.10)}, this last equation implies $X^k \partial_ku^k=0$ (no summation over the index $k$).
Up to a reordering of the coordinates, let us assume that $X^a=0$ and $X^\alpha\neq 0$ for $a=1,\ldots, m$ and $\alpha=m+1,\ldots, n$. From the last equation it follows that $u^\alpha={\rm const}$ and $\bs X=X^\alpha \partial_\alpha$. At the beginning of this proof we have seen that if $u^\alpha={\rm const}$, then $q^\alpha$ are ignorable coordinates (we always assume that the coordinates $(q^i)$ are normalised so that the metric assumes  the form~\eqref{(7.3)}). Thus, $\partial_\alpha$ are Killing vectors in involution and eigenvectors of~$\bs L$. Since $\bs X=X^\alpha \partial_\alpha$ is a~Killing vector, the components~$X^a$ must be constant.
\end{proof}

Assume that the eigenvalues $u^\alpha$ ($u^\alpha=m+1,\ldots,n$) of an L-tensor are constant and the remaining~$(u^a)$ ($a=1,\ldots,m$) are independent functions. Then we can choose associated orthogonal coordinates $(q^i)=(q^a,q^\alpha)$ such that $q^a=u^a$ and $q^\alpha$ are ignorable. From~\eqref{(7.1)} it follows that
\begin{gather*}
\big(u^a-u^b\big) \partial_a\ln g^{bb}+1=0, \qquad
\big(u^a-u^\alpha\big) \partial_a\ln g^{\alpha\alpha}+1=0,
\end{gather*}
being the remaining equations identically satisf\/ied. Thus, we are faced with three cases:
\begin{enumerate}\itemsep=0pt
\item[(I)] $m=0$, all $u^i={\rm const}$, i.e., all $q^i$ are ignorable: the manifold $Q$ is locally f\/lat, the coordinates $(q^i)$
are orthogonal Cartesian coordinates, $g^{ii}={\rm const}$ and $\bs L$ is a constant tensor.

\item[(II)] $0<m<n$: in this case $g^{\alpha\alpha}\neq{\rm const}$ due to equation $(u^a-u^\alpha)\partial_a\ln g^{\alpha\alpha}+1=0$.
Condition $g^{\alpha\alpha}\neq{\rm const}$ means that the Killing vectors $\bs X_\alpha=\partial_\alpha$ are  not translations \cite[Section~52]{Eisenhart-1933}.

\item[(III)] $m=n$, all eigenvalues are independent: this is the case examined in Theorem~\ref{t:9.2}.
\end{enumerate}

Taking into account the proof of the preceding theorem, cases (I) and (II) shows that:

\begin{Theorem}\label{t:9.3}
Let $\bs L$ be an L-tensor.
$(i)$ If $\bs L$ has all constant eigenvalues, then the mani\-fold~$Q$ is locally flat and $\bs L={\rm const}$ $($in the sense that all its components in Cartesian coordinates are constant$)$.
$(ii)$ If $\bs L$ is invariant with respect to  $m<n$ Killing vectors, then these vectors are not translations.
\end{Theorem}

\section{L-potentials}\label{section10}

In this section we propose a few remarks on the {\it potential functions} associated with an L-system. For further approaches to this matter we refer the reader to \cite{Blaszak-Ma-2003, Ibort-Magri-Marmo-2000}. By virtue of~\eqref{(1.3)} and~\eqref{(7.3)} we have

\begin{Theorem}
A potential $V$ is separable in an L-system if and only if, with respect to  normal coordinates $(q^i)$ associated with $\bs L$,
\begin{gather*}
\partial_i\partial_jV=\dfrac 1{u^j-u^i} (\partial_jV-\partial_iV),\qquad  \partial_i=\dfrac{ \partial}{\partial q^i}.
\end{gather*}
\end{Theorem}

Note that in this theorem (as well as in the following) the eigenvalues $u^i$ of $\bs L$ may not be independent functions.

\begin{Theorem}
Let $(\bs K_a)$ be the basis of the KS-space generated by an L-tensor $\bs L$ according to formula~\eqref{(8.7)}. Then the functions $V_a(\underline u)\doteq \sigma_{a+1}(\underline u)$ $(a=0,1,\ldots, n-1)$ fulfill equations
$dV_a=\bs K_a\,dV$, with $V=V_0=\operatorname{tr}(\bs L)=\sum_iu^i$, and
the functions $H_a=\tfrac 12 P_{\bs K_a}+V_a$ are first integrals in involution.
\end{Theorem}

\begin{proof}
$\bs K_a\,dV=\sigma_a^i \partial_iV\,dq^i=\sigma_a^i\,dq^i=\partial_i\sigma_{a+1}\,dq^i=dV_a$. Apply~\eqref{(5.1)}.
\end{proof}

\section{L-pencils}\label{section11}

There are St\"ackel webs which are not L-webs. A necessary condition is the Robertson condition. However, also in manifolds where this condition is identically satisf\/ied, for instance in constant curvature spaces, there are St\"ackel systems which are not generated by an L-tensor. For instance, due to Theorem~\ref{t:9.3}(ii), in a~$\mathbb E_n$ all translational webs (except the Cartesian web) are not L-webs. As remarked above, also the spherical-conical webs are not L-webs. This last case has suggested~\cite{Benenti-1992b} the introduction of a linear combination\footnote{See also \cite{Benenti-2008}.\label{pagenote3}}
\begin{gather*}
\bs L(m)=\bs L_0+m\bs L_1,
\end{gather*}
which is an L-tensor for all values the parameter $m\in \mathbb R$. We call {\it L-pencil} such an object.

\begin{Theorem}
Let $\bs L=\bs L_0+m\bs L_1$ be an L-pencil. Then:
\begin{enumerate}\itemsep=0pt
\item[$(i)$] $\bs L_0$ has simple eigenvalues.
 \item[$(ii)$] $\bs L$~is a CKT for all $m$ if and only if $\bs L_0$ and $\bs L_1$ are CKT's.

\item[$(iii)$] The condition $\bs H(\bs L)=0$ is equivalent to
 \begin{gather*}
\bs H(\bs L_0)=\bs H(\bs L_1)=0, \qquad
\underset{0}{L}{\phantom |}_{[i}^h \partial_{|h|} \underset{1}{L}{\phantom |}_{j]}^k
-\underset{0}{L}{\phantom |}_h^k \partial_{[i} \underset{1}{L}{\phantom |}_{j]}^h
+
\underset{1}{L}{\phantom |}_{[i}^h \partial_{|h|} \underset{0}{L}{\phantom |}_{j]}^k
-\underset{1}{L}{\phantom |}_h^k \partial_{[i} \underset{0}{L}{\phantom |}_{j]}^h=0.
\end{gather*}
\item[$(iv)$]  The Ricci tensor $\bs R$ commutes with both $\bs L_0$ and $\bs L_1$.
\end{enumerate}
\end{Theorem}

The proof is straightforward. Let us apply the iterative formula~\eqref{(8.7)} and Theorem~\ref{t:8.1} to the tensor $\bs L(m)$. Since $\bs L$ is polynomial of degree~1 in $m$, each tensor $\bs K_a(m)$ is at most of degree~$a$:
$\bs K_a(m)=\bs H_a m^a+\bs H_{a-1}m^{a-1}+\cdots$.

\begin{Theorem}
$(i)$ The tensors $\bs K_a(0)$ form a KS-space.
$(ii)$ The tensors $\bs H_a$, given by the coefficients of maximal degree of $\bs K_a(m)$, if independent, form a KS-space. $(iii)$ The St\"ackel systems generated by these two KS-spaces satisfy the Robertson condition.
\end{Theorem}

Also for this statement we omit the proof. We remark that there are cases in which all the~$\bs K_a$ are of degree~1, as in the following example.

\begin{Example}[\cite{Benenti-1992b}] In $\mathbb E_n$ we consider the tensor $\bs L(m)=\bs L_0+m \bs r\otimes\bs r$. It represents (for $m\neq 0$) the planar moment of inertia of a massive body with total mass $m$ and center of mass at the origin. $\bs L_0$ is a constant symmetric tensor (hence, a~KT) with simple eigenvalues~$(a^i)$. $\bs L_1=\bs r\otimes \bs r$ is a CKT whose eigenvalues are all zero except one ($={}r^2$). It can be proved that: (i)~$\bs L$ is an L-pencil; (ii)~the tensors $\bs K_a(0)$ form a KS-space corresponding to Cartesian coordinates; (iii)~$\bs K_a(m)$ are all of degree~1 in $m$, $\bs K_a(m)=\bs K_a(0)+m \bs H_a$; (iv)~the tensors $\bs H_a$ form a KS-space corresponding to the conical spherical coordinates.
\end{Example}

The above results stimulate investigations about the notion of an L-pencil, also in relation with recent studies on the same concept and the notion of {\it cofactor pair system} (see, e.g., \cite{Lundmark-2001, Marciniak-Blaszak-2002, Rauch-Wojciechowski-Marciniak-Lundmark-1999,
Rauch-Wojciechowski-Waksjo-2003}). A possible generalisation is a~{\it multi-pencil} of the type $\bs L=\bs L_0+m^i\bs L_i$ with $(m^i)\in \mathbb R^k$ or $\bs L=m^i\bs L_i$ with $(m^i)\neq \bs 0$. Basic examples are, in the Euclidean $n$-space, the case $\bs L=\bs L_0+ m  \bs r\otimes\bs r+ w \bs c\odot \bs r$ where $\bs c$ is a constant unit vector, see~\eqref{(4.2)}.
For brevity, we do not examine here which St\"ackel webs in $\mathbb E_n$ (or in~$\mathbb S_n$, $\mathbb H_n$) are generated by an L-tensor or by an L-pencil.

\section{The Riemannian background of the separation}\label{section12}

Also for a better understanding of the orthogonal separation, it is necessary to study the Levi-Civita equations without any {\it a priori} assumptions on the separable coordinates $\underline q$. The geometrical meaning of these equations has been f\/irstly investigated by Kalnins and Miller~\cite{KM-1980,KM-1981} and in~\cite{Benenti-1980} (see also~\cite{Benenti-1991,Benenti-1996}). A result of these investigations is the following.

\begin{Theorem}[\cite{Benenti-1997}]\label{t:12.1}
The Hamilton--Jacobi equation $G+V=E$ is separable if and only if there exists a characteristic Killing pair $(D,\bs K)$ such that
\begin{gather}
DV=0, \qquad d(\bs K \,dV)=0.
\label{(12.1)}
\end{gather}
\end{Theorem}

A {\it characteristic Killing pair} $(D,\bs K)$ is made of a $r$-dimensional linear space $D$ of commuting Killing vectors and a $D$-invariant Killing tensor $\bs K$ having $m=n-r$ normal eigenvectors orthogonal to $D$ and corresponding to distinct eigenvalues (note that these eigenvalues may not be simple). Since the elements of $D$ are in involution, they generate $r$-dimensional orbits which are locally f\/lat submanifolds. It can be shown that {\it if $(D,\bs K)$ is a characteristic Killing pair, then $D$ is normal}, in the sense that the distribution orthogonal to its elements is completely integrable. This means that there is a foliation of $m$-dimensional manifolds orthogonal to the orbits of~$D$. The leaves of this foliation are isometric Riemannian manifolds (the isometries being generated by the orbits of $D$, assuming that they intersect the orbits in only one point).

With a characteristic Killing pair we associate {\it standard separable coordinates} $(q^i)=(q^a,q^\alpha)$
($a=1,\ldots,m$, $\alpha=m+1,\ldots, n$) def\/ined in this way:
(i)~$dq^a$ are eigenforms of~$\bs K$ orthogonal to~$D$ (or, equivalently, the orbits of $D$ are def\/ined by equations $q^a={\rm const}$); (ii)~the $(q^\alpha)$ are the af\/f\/ine parameters, with zero value on an arbitrary $m$-dimensional orthogonal section $\cal Z$  of the orbits of~$D$, of Killing vectors $\bs X_\alpha$ forming a local basis of $D$. It follows that: (I)~$\partial/\partial q^\alpha=\bs X_\alpha$ and the coordinates $(q^\alpha)$ are ignorable;
(II)~in these coordinates the contravariant metric tensor components assume the {\it standard form}
\begin{gather}
\big[g^{ij}\big]=
\bmatrix \big[g^{aa}\big] & [0]
\\[3pt]
 [0] & \big[g^{\alpha\beta}\big] \cr \endbmatrix,
\label{(12.2)}
\end{gather}
where $\big[g^{aa}\big]$ is a diagonal $m\times m$ matrix and $\big[ g^{\alpha\beta}\big]$ is a $r\times r$ matrix, with $m=n-r$. Here we denote by $[0]$ zero-matrices of proper dimensions. The rather long proof of this theorem is based on the following

\begin{Theorem}[\cite{Benenti-1980}]\label{t:12.2}
In an equivalence class of separable coordinates there exist coordinates in which the metric assume the standard form~\eqref{(12.2)}.
\end{Theorem}

Two separable systems are called {\it equivalent} if (in the intersection of their domains of def\/i\-ni\-tion) the corresponding separated solutions of the Hamilton--Jacobi equation generate the same Lagrangian foliation of~$T^*Q$.
The geometrical representation of an equivalent class is given by
a {\it separable web}: it is a pair $(D,{\cal S}^a)=(D, ({\cal S}^1,\ldots,{\cal S}^m))$, where $D$ is a normal $r$-dimensional space of commuting  KT's and ${(\cal S}^a)$ is a family of $m$ orthogonal foliations
of submanifolds of codimension~1, all invariant with respect to  $D$ (it follows that the orbits of $D$ are the complete intersections of the $({\cal S}^a)$). Moreover, the submanifolds $({\cal S}^a)$ are orthogonal to $m$ eigenvectors of a $D$-invariant KT $\bs K$ with distinct (but not necessarily simple) eigenvalues (so that $(D,\bs K)$ is a characteristic Killing pair).
It must be observed that the quotient set of the orbits of $D$ is a $m$-dimensional manifold with a naturally induced orthogonal separable metric $(g^{aa})$ (it is isomorphic to any $m$-manifold orthogonal to~$D$).

In this description of the separation we include the extreme cases: (I)~$m=n$, $r=0$, which corresponds to the an orthogonal separation; (II)~$m=0$, $r=n$, which corresponds to the separation in Cartesian coordinates (in this case the manifold is f\/lat). In \cite{BCR-2002a} it is proved that

\begin{Theorem}\label{t:12.3}
A characteristic Killing pair $(D,\bs K)$ generates a $m$-dimensional space ${\cal K}_m$ of Killing tensors with the following properties:
$(i)$~they are $D$-invariant and have $m$ eigenvectors in common orthogonal to $D$;
$(ii)$~they are in involution;
$(iii)$~if the  characteristic equations~\eqref{(12.1)} of a separable potential $V$ are satisfied, then they are satisfied by all elements of~${\cal K}_m$.
\end{Theorem}

\begin{Theorem}
In a space ${\cal K}_m$ having properties $(i)$ and $(ii)$, where $D$ is a $n-m$-space of Killing vectors in involution, all the common eigenvectors are normal.
\end{Theorem}

\begin{Theorem}
The functions $H_a=\tfrac 12  P_{\bs K_a}+V_a$ and $P_{\bs X_\alpha}=X_\alpha^i p_i$, where $dV_a=\bs K_a\,dV$ and~$(\bs X_\alpha)$ is a basis of~$D$, form a~complete system of integrals in involution.
\end{Theorem}

Theorems~\ref{t:12.2} and~\ref{t:12.3} show that if $D$ admits an orthogonal basis, $\bs X_\alpha{\mathbf\cdot}\bs X_\beta=0$ for $\alpha\neq\beta$, then the standard coordinates are orthogonal. This means that the separation occurs in orthogonal coordinates. It can be proved by a coordinate-independent method that

\begin{Theorem}[\cite{Benenti-1992a}]
On a manifold with constant curvature any normal space $D$ of Killing vectors in involution has an orthogonal basis.
\end{Theorem}

As a consequence,

\begin{Theorem}
On a manifold of constant curvature the geodesic separation always occurs in orthogonal coordinates.
\end{Theorem}

This important property was discovered by Kalnins and Miller \cite{KM-1984a, KM-1986} (see also \cite{Kalnins-1986}), within a coordinate-dependent approach. This theorem can be extended to a natural Hamiltonian $H=G+V$.\footnote{The study of the Riemannian background of the separation has been extended
to a~Hamiltonian with scalar and vector potential in~\cite{BCR-2001}.}

\pdfbookmark[1]{References}{ref}
\LastPageEnding

\end{document}